\documentclass[12pt]{article}

\usepackage{latexsym}
\usepackage{comment}
\excludecomment{workahead}

\newtheorem{theorem}{Theorem}
\newtheorem{definition}{Definition}
\newtheorem{corollary}{Corollary}

\newtheorem{lemma}{Lemma}

\newcommand{\cR}{\mbox{${\cal R}$}}

\newcommand{\blackslug}{\mbox{\hskip 1pt \vrule width 4pt height 8pt 
depth 1.5pt \hskip 1pt}}
\newcommand{\qed}{\quad\blackslug\lower 8.5pt\null\par\noindent}
\newenvironment{proof}{\par\noindent{\bf Proof:}}{\qed \par}

\begin{document}

\title{Ultra valuations}
\author{Daniel Lehmann \\
School of Engineering and Computer Science, \\
Hebrew University, \\
Jerusalem 91904, Israel \\
daniel.lehmann1@mail.huji.ac.il }
\maketitle
\begin{abstract}
This paper considers an original exchange property of set valuations.
This property is shown to be equivalent to a property described in~\cite{DT:AML} 
in the context of discrete optimization and matroids and shown there to characterize
the valuations for which the demand oracle can be implemented by a greedy algorithm.
The same exchange property is also equivalent to a property described independently 
 in~\cite{RGP:ET} and in~\cite{LLN:GEB} and shown there 
to be satisfied by substitutes valuations.
It is also equivalent to an ultra-metric property 
of the complementarity exhibited by a valuation.
The paper then studies the family of valuations that satisfy this exchange property, 
the ultra valuations.
Any substitutes valuation is an ultra valuation, but ultra valuations may exhibit 
complementarities.
Ultra valuations satisfy the law of aggregate demand introduced 
in~\cite{Hatfield_Milgrom:2005}.
Any symmetric valuation is an ultra valuation.
Substitutes valuations are exactly the submodular ultra valuations.
Ultra valuations define ultrametrics on the set of items.
The maximum of an ultra valuation on $n$ items can be found in $O(n^2)$ steps.
\end{abstract}


\section{Introduction} \label{sec:intro}
Two different streams of research: equilibrium theory in Economics 
and discrete optimization in Applied Mathematics have
considered closely related families of set functions but used different languages, 
notations and auxiliary assumptions common in their fields.
They sometimes failed to realize that the families they considered were essentially the same.

Kelso and Crawford~\cite{KelsoCraw:82} introduced the (gross) substitutes valuations
in order to model the work market. 
They showed that a market in which all agents have substitutes valuations 
has a Walrasian equilibrium.
Since then, the study of substitutes valuations and of discrete markets amongst  agents 
exhibiting substitutes valuations have been intensively pursued, 
see in particular Bikhchandani and Mamer~\cite{BikhMamer:equ},
Gul and Stacchetti~\cite{GulStacc:99} and Reijnierse, van Gellekom and 
Potters~\cite{RGP:ET}.
In those works one mostly considers properties of valuations under different item prices.

Applied mathematicians studied discrete convex optimization, matroids and greedy algorithms,
see in particular Jensen and Korte~\cite{JensenKorte:82}, Korte, Lov\'{a}sz and 
Schrader~\cite{Korte_Lovasz_Schrader:91}, Dress and Terhalle~\cite{DT:AML} 
and Murota and Shioura~\cite{MurotaShioura:99}.
In those works one mostly considers exchange properties that do not involve prices.

The connection between those two streams was noticed by 
Fujishige and Yang~\cite{FujiYang:03} and
Lehmann, Lehmann and Nisan~\cite{LLN:GEB} 
who considered more general families of valuations exhibiting no complementarity 
and extended some of their results to valuations with limited complementarities.
The algorithmic aspects of the topic have been remarkably surveyed in~\cite{RPL:GEB}
which also presents original results.

The purpose of this paper is to present an original exchange property 
and show that it is equivalent to a number of properties studied 
in~\cite{RGP:ET, DT:AML, LLN:GEB, Bing_Lehmann_Milgrom:Leibniz} 
and that it implies the law of aggregate demand of~\cite{Hatfield_Milgrom:2005}.
Valuations that satisfy this exchange property will be called ultra valuations, after one of 
those equivalent properties, related to ultra-metrics and presented in Section~\ref{sec:ultra}
below.

This exchange property greatly facilitates the proof of all properties of ultra valuations and
in particular the fact that a greedy algorithm finds an optimal bundle under any price vector.
This greedy algorithm has been described by Dress and Terhalle in~\cite{DT:AML} 
where the class of valuations for which it finds a maximum is characterized 
by the exchange property in Definition~\ref{def:DT} below. 
The authors did not relate their property with the substitutes property.
This paper shows that the class of valuations that satisfy Dress and Terhalle's property is
strictly larger than the class of substitutes valuations.
This greedy algorithm is not the one described in~\cite{RPL:GEB} and shown there to 
characterize substitutes valuations.
Our exchange property is also useful in showing that substitutes valuations 
are exactly the submodular ultra valuations and in looking for a notion of equilibrium
suitable for exchanges in which agents exhibit ultra valuations.
This allows for a characterization of the substitutable choice functions of Hatfield and 
Milgrom's~\cite{Hatfield_Milgrom:2005} in Appendix~\ref{app:CL}.
In Milgram and Hatfield's model the set of items has a structure: 
the items are contracts between a hospital and a doctor.
This structure is used in a number of later works, 
including~\cite{Hatfield_Kominers:AEJ, Hatfield_Kojima:comment, Hatfield_Kominers:hidden, 
Hatfield_Kominers:GEB, Hatfield+4:2018}, where, for example, 
items (contracts) which refer to the same doctor are given special consideration.
The many classes of valuations considered in those works use, in their definition, 
this extra structure. 
Therefore they cannot be compared with the classes of the present work which assumes that
the set of items is unstructured.

\section{Basic notions and notations } \label{sec:basic}
We consider a finite set of items. 
We follow~\cite{GulStacc:99} and name this set $\Omega$ and let \mbox{$n \: = \:$} 
\mbox{$\mid \Omega \mid$}.
A subset of $\Omega$ is a bundle. A valuation gives a real value to every bundle:
any real function \mbox{$v : {2}^{\Omega} \longrightarrow \cR$} is a valuation.
Note that we do not require $v$ to be monotonic, non-negative or that \mbox{$v(\emptyset)$}
be equal to $0$.

In the sequel $A$, $B$, $X$, $Y$, \ldots will always denote subsets of 
$\Omega$, 
and $a$, $b$, $x$, $y$, \dots will denote elements of $\Omega$. 
The number of elements of $A$ is denoted $\mid A \mid$.
The set \mbox{$A \cup \{ x \}$} will
be denoted \mbox{$A + x$} {\em only when} \mbox{$x \not \in A$} 
and the set \mbox{$A - \{ x \}$} will be denoted
\mbox{$A - x$} {\em only when} \mbox{$x \in A$}. 
The set \mbox{$A - x + y$} will denote \mbox{$( A - \{ x \}) \cup \{ y \}$} {\em only
when} \mbox{$x \neq y$}, \mbox{$x \in A$} and \mbox{$y \not \in A$}. 
The set \mbox{$\{ x \} \cup \{ y \}$} will be denoted
\mbox{$x+y$} {\em only when} \mbox{$x \neq y$}.
If \mbox{$X \subseteq \Omega$} we shall denote by $v^{X}$ the restriction
of $v$ to subsets of $X$ defined by: \mbox{$v^{X}(A) = v(A)$} for all
\mbox{$A \subseteq X$}. 
The marginal valuation defined by $v$ given a bundle 
\mbox{$A \subseteq \Omega$}, denoted $v_{A}$ is a valuation on $\Omega - A$ defined by
\[
v_{A}(B) = v(A \cup B) - v(A)
\]
for any bundle \mbox{$B \subseteq \Omega - A$}.
For convenience, the expression \mbox{$v( x \mid y )$} will be used to denote 
\mbox{$v_{\{ y \}} ( \{ x \} )$} and \mbox{$v_{A}(x \mid B)$} for
\mbox{${(v_{A})}_{B}(x) \: = \:$} \mbox{$v_{A \cup B}(x)$}. 
We shall {\em not} use \mbox{$v(x \mid B \mid A )$}.

We shall now define a measure of the degree of complementarity between two items.
It will be a central technical tool in the sequel.
\begin{definition} \label{def:c}
Let $v$ be a valuation on $\Omega$, \mbox{$A \subseteq \Omega$} 
and let $x$, \mbox{$y \in \Omega - A$} be distinct items.
The complementarity of $x$ and $y$ given $A$ is defined as:
\[
c_{A} ( x , y ) \: = \: v_{A} (x + y) - v_{A} ( x ) - v_{A} (y).
\]
\end{definition}
Note that \mbox{$c_{A} ( y , x ) \: = \: c_{A} ( x , y )$} and that if $x$ and $y$ are
complementary in the presence of $A$ the quantity \mbox{$c_{A}(x , y)$} is positive and that
it is negative if $x$ and $y$ are substitutes in the presence of $A$.

\begin{definition} \label{def:additive}
A valuation \mbox{$p : {2}^{\Omega} \rightarrow \cR$} is said to be {\em additive} iff
\mbox{$f(\emptyset) = 0$} and
for any \mbox{$A , B \subseteq \Omega$} one has
\mbox{$p(A) + p(B) \: = \:$} \mbox{$p(A \cup B) + p(A \cap B)$}.
\end{definition}
Often an additive valuation $p$ presents a price structure in which every bundle is valued
at the sum of the prices of the items it contains.

\begin{definition} \label{def:utility}
If $v$ is any valuation and $p$ is an {\em additive} valuation, the valuation $v^{p}$ defined by
\mbox{$v^{p} ( A ) \, = \,$} 
\mbox{$v(A) - p(A)$} is the {\em utility} defined by $v$ and $p$.
\end{definition}

\section{Plan of this paper} \label{sec:plan}
In Section~\ref{sec:u-val} an original exchange property, {\bf Exchange}, 
that defines ultra valuations is presented and some basic facts established.
In order to show immediately the relevance of this property to economics and auctions
Section~\ref{sec:aggregate} shows
that any ultra valuation satisfies the law of aggregate demand 
of~\cite{Hatfield_Milgrom:2005}.
Section~\ref{sec:ultra} presents the main technical result of this paper.
It describes the {\bf Ultra} property and shows that it implies {\bf Exchange}.
The {\bf Ultra} property is inspired by~\cite{Bing_Lehmann_Milgrom:Leibniz}.
Section~\ref{sec:LLN} shows that properties, {\bf LLN} and {\bf RGP}, respectively studied 
by~\cite{LLN:GEB} and~\cite{RGP:ET} are trivially equivalent and
imply {\bf Ultra}.
Section~\ref{sec:DT} shows that a property studied by~\cite{DT:AML} implies {\bf LLN}
and is implied by {\bf Exchange}.
One concludes that all properties mentioned above are equivalent.
Section~\ref{sec:examples} provides examples of ultra valuations
and Section~\ref{sec:closure} studies some closure properties 
of the family.
Section~\ref{sec:ultra_sub} characterizes substitutes valuations as those
ultra valuations that are submodular. 
The proof is an interesting alternative to~\cite{RGP:ET}'s proof.
Section~\ref{sec:preferred} studies, for ultra valuations, 
the structure of preferred bundles, i.e., those bundles that maximize the valuation. 
Some of the properties put in evidence seem novel even for substitutes valuations.
Section~\ref{sec:prices} studies the changes in preferred bundles caused by a change in 
the prices of items.
Section~\ref{sec:optimal} proposes some characterizations of preferred bundles.
Section~\ref{sec:greedy} provides an alternative proof to~\cite{DT:AML}'s result:
ultra valuations are exactly those for which a greedy algorithm finds a bundle of maximal value.
Section~\ref{sec:Walras} characterizes competitive (Walrasian) equilibria among agents
exhibiting ultra valuations and discusses the kind of stability that can be expected when
there is no competitive equilibrium.
Section~\ref{sec:conclusion} concludes this paper and discusses open problems 
and future work.
Appendices~\ref{app:lemmas}, \ref{app:dummy} and~\ref{app:substitutes}
contain technical lemmas used in the proofs.
Appendix~\ref{app:CL} studies the relation between ultra valuations and
the substitutability property of~\cite{Hatfield_Milgrom:2005} couched in choice-language.

\section{Ultra valuations} \label{sec:u-val}

Exchange properties are an important tool in discrete optimization 
and a number of such properties have been considered in the literature,
notably by Dress and Terhalle in~\cite{DT:AML} (see Definition~\ref{def:DT}) and by
Murota and Shioura in~\cite{MurotaShioura:99} (see Definition~\ref{def:M-natural}).
We shall now introduce an original exchange property.
This property can, without too much effort be shown to imply the former 
and be implied by the latter. 
We shall show, with substantial effort, that it is equivalent to the former and strictly weaker
than the latter.
Our exchange property has no obvious intuitive meaning: one cannot, on first principles, 
characterize the type of economic agents whose valuation satisfies it, but neither is this the case for
the properties just mentioned. 
Note that, in the definition below, to any \mbox{$x \in A - B$} must correspond 
a \mbox{$y \in B - A$}, as in the $M^{\natural}$-concavity of~\cite{MurotaShioura:99}.

\begin{definition} \label{def:u-val}
Let $\Omega$ be a finite set and \mbox{$v : {2}^{\Omega} \rightarrow \cR$} be 
a valuation on $\Omega$.
The valuation $v$ is an {\em ultra} valuation iff
for any \mbox{$A , B \subseteq \Omega$} such that \mbox{$\mid A \mid \leq \mid B \mid$}
and any \mbox{$x \in A - B$} there is some \mbox{$y \in B - A$} such that
\[ 
{\bf Exchange} \ \ v(A) + v(B) \: \leq \: v(A - x + y) + v(B - y + x).
\]
\end{definition}

The reason for the term {\em ultra} will become clear in Section~\ref{sec:ultra}.
A striking feature of {\bf Exchange} is the requirement that the size of the bundle $A$, 
from which $x$ will be taken contains is not greater than that of $B$.
It may seem odd that size should be a consideration.
In Section~\ref{sec:aggregate} the reader will meet a property considered 
in~\cite{Hatfield_Milgrom:2005}, with a solid economic interpretation, 
that also implies size considerations.

But, first, let us state some simple facts.
\begin{lemma} \label{the:marginal}
If $v$ is an ultra valuation then, 
\begin{enumerate}
\item for any \mbox{$S \subseteq \Omega$} the
marginal valuation $v_{S}$ is also an ultra valuation, and
\item for any additive valuation $p$, the valuations \mbox{$v + p$} 
and \mbox{$v - p$} are ultra valuations.
\end{enumerate}
\end{lemma}
\begin{proof}
For the first claim, we have \mbox{$\mid A \mid \: \leq \:$} \mbox{$ \mid B \mid$} iff 
\mbox{$\mid A \cup S \mid \: \leq \: $} \mbox{$\mid B \cup S \mid$},
\mbox{$ A - B \: = \: $} \mbox{$ (A \cup S) - (B \cup S)$},
and
\[
v_{S} ( A ) + v_{S} ( B ) \: \geq \: v_{S} ( A - x + y ) + v_{S} (B - y + x ) {\rm \ \  iff}
\]
\[
 v ( A \cup S ) + v ( B \cup S) \: \geq \: v ( A \cup S - x + y ) + v (B \cup S - y + x ).
\]
Secondly, 
\[
\sum_{z \in A} p_{z} + \sum_{z \in B} p_{z} \: = \:
\sum_{z \in A - x + y} p_{z} + \sum_{z \in B - y + z} p_{z}.
\]
\end{proof}

\section{The law of aggregate demand} \label{sec:aggregate}
In~\cite{Hatfield_Milgrom:2005} Hatfield and Milgrom developed a general model
that encompasses many situations previously studied in the literature.
They identified a common property of those many situations: the law of
aggregate demand that, informally, says that, if the set of possible choices is enlarged,
the number of items desired cannot decrease.
They formulate their condition under the assumption that, for any \mbox{$X \subseteq \Omega$}
there is a unique bundle maximizing $v$ among all subsets of $X$.
Definition~\ref{def:aggregate} reduces to their definition in such a case, but treats properly
the general case with ties.
\begin{definition} \label{def:aggregate}
A valuation $v$ satisfies the {\em law of aggregate demand} iff for any partition of the set
of items \mbox{$\Omega \: = \: X \cup Y$}, \mbox{$X \cap Y \: = \: \emptyset$}
and for any bundle \mbox{$A \subseteq X$} that maximizes $v$ among all subsets
of $X$ there is some bundle \mbox{$B \subseteq \Omega$} that maximizes $v$ over
all subsets of $\Omega$ such that \mbox{$\mid B \mid \: \geq \:$} 
\mbox{$\mid A \mid$}. 
\end{definition}

\begin{theorem} \label{the:aggregate}
Any ultra valuation satisfies the law of aggregate demand.
\end{theorem}
The converse does not hold.
\begin{proof}
Let $v$ be an ultra valuation and let \mbox{$\Omega \: = \:$} \mbox{$X \cup \{ z \}$}
for \mbox{$z \in \Omega - X$}.
Assume \mbox{$A \subseteq X$} maximizes $v$ among all subsets of $X$.
Let \mbox{$Q \subseteq \Omega$} be any bundle that maximizes $v$ among all subsets 
of $\Omega$.
We want to prove that there is some \mbox{$B \subseteq \Omega$} such that 
\mbox{$v(B) \: = \: v(Q)$} and \mbox{$\mid B \mid \: \geq \:$} \mbox{$\mid A \mid$}.
Without loss of generality we may assume that \mbox{$\mid Q \mid \: < \:$}
\mbox{$\mid A \mid$} and also that \mbox{$v(Q) \: > \:$} \mbox{$v(A)$}, 
implying \mbox{$z \in Q - A$}.
{\bf Exchange} then implies that there is some \mbox{$y \in A - Q$}
such that
\[
v(Q) + v(A) \: \leq \: v(Q - z + y) + v(A - y + z).
\]
But \mbox{$v(Q) \: \geq \:$} \mbox{$ v(A - y + z)$} since $Q$ maximizes $v$ over all subsets and
\mbox{$v(A) \: \geq \:$} \mbox{$v(Q - z + y)$} since \mbox{$Q - z + y \subseteq X$}.
We conclude that \mbox{$v(Q) \: = \:$} \mbox{$ v(A - y + z)$} and we can take
\mbox{$B \: = \:$} \mbox{$A - y + z$}.
The proof is concluded by adding to $X$ the elements of $Y$ one by one.
\end{proof}

\section{The Ultra property} \label{sec:ultra}
The following originates in~\cite{Bing_Lehmann_Milgrom:Leibniz}.
\begin{definition} \label{def:ultra}
Let $v$ be any valuation on $\Omega$.
If for any \mbox{$A \subseteq \Omega$}
and any pairwise distinct \mbox{$x , y , z \in \Omega - A $} 
such that \mbox{$c_{A} ( x , y ) \: > \:$}
\mbox{$c_{A} ( x , z )$} one has \mbox{$c_{A} ( y , z )\:  = \:$}
\mbox{$c_{A} ( x , y )$},
we shall say that $v$ satisfies the {\bf Ultra} property.
\end{definition}
The property {\bf Ultra} is the characteristic property of ultra-metrics: 
all triangles are isoceles and the equal sides are the longer ones.
This makes $c_{A}$ almost an ultra-metric, but notice that 
\mbox{$c_{A} ( x , y )$} may be negative. 

Lemma~\ref{the:central} is the hard core of this paper.
\begin{lemma} \label{the:central}
If $v$ satisfies the {\bf Ultra} property then $v$ is an ultra valuation, 
i.e., it satisfies the {\bf Exchange} property.
\end{lemma}
In Theorem~\ref{the:equivalence} we shall conclude that the converse holds: 
{\bf Exchange} implies {\bf Ultra} and this is much easier to prove.
\begin{proof}
We shall rely on a number of lemmas the proof of which appears in
Appendix~\ref{app:lemmas}.
Assume $v$ satisfies {\bf Ultra}.
Let \mbox{$A , B \subseteq \Omega$}, \mbox{$\mid A \mid \: \leq \:$}
\mbox{$ \mid B \mid$} and \mbox{$x \in A - B$}.
Let \mbox{$X = A \cap B$}.
Note that $X$ may well be empty, in this case \mbox{$v_{X} \: = \: v$}.
Our goal is equivalent to showing that there is some \mbox{$y \in B - A$} such that
\[
v_{X} ( A - B ) + v_{X} ( B - A ) \: \leq \: v_{X}(A - B - x + y) + v_{X}(B - A - y + x ).
\]
But, now, \mbox{$(A - B) \cap (B - A) \: = \:$} $\emptyset$ and Lemma~\ref{le:conditional} 
shows that $v_{X}$ satisfies {\bf Ultra}.
It is therefore enough for us to show that {\bf Exchange} holds when 
\mbox{$A \cap B \: = \:$} $\emptyset$. 
We assume \mbox{$A \cap B \: = \:$} $\emptyset$.

We shall prove our claim by induction on the size of \mbox{$A$}.
Since \mbox{$x \in A$} this set is not empty and our base case is
\mbox{$\mid A \mid \: = \:$} $1$, i.e., \mbox{$A \: = \:$} \mbox{$\{ x \}$}.
Since \mbox{$\mid A \mid \: \leq \: $} \mbox{$\mid B \mid$}, \mbox{$B \neq \emptyset$}.
We want to show that there is some \mbox{$y \in B$} such that
\[
v(x) + v(B) \: \leq \: v(y) + v(B - y + x).
\]
This is proved in Lemma~\ref{le:1-n}, in Appendix~\ref{app:lemmas}.

Suppose now that {\bf Exchange} (for \mbox{$A \cap B \: = \:$} $\emptyset$) 
holds if \mbox{$\mid A \mid \: \leq \:$} \mbox{$ n - 1$}
(\mbox{$n \: \geq \:$} $2$) and assume \mbox{$\mid A \mid \: = \:$} $n$.
Let $w$ be any element of \mbox{$A - x$} and
let $z$ be the element of \mbox{$B$} that maximizes the quantity
\mbox{$v(B - z' + w) - v(B - z' + x)$} over all \mbox{$z' \in B$}.
By the induction hypothesis there is some \mbox{$y \in B - z$} such that
\begin{equation} \label{eq:one}
v_{w}(A - w) + v_{w}(B - z) \: \leq \: v_{w}(A - w - x + y) + v_{w}(B - z - y + x).
\end{equation}
We have
\begin{equation}
v(A) + v(B - z + w) \: \leq \: v(A - x + y) + v(B - z - y + w + x).
\end{equation}
Let \mbox{$X = B - y - z$}
By Lemma~\ref{le:2-2} and Lemma~\ref{le:conditional}, in Appendix~\ref{app:lemmas},
\[
v_{X} ( y + z ) + v_{X} ( w + x) \: \leq \: \max(v_{X} (x + y) + v_{X}(w + z) , 
v_{X}(w + y) + v_{X}(x + z) ).
\]
and therefore
\[
v(B) + v(B - y - z + w + x) \: \leq \: 
\]
\[
\max(v(B - z + x) + v(B - y + w) , 
v(B - z + w) + v(B - y + x) ).
\]
By the choice of $z$ we have
\[
v(B - z + w) - v(B - z + x ) \: \geq \: v(B - y + w) - v(B - y + x) 
\]
and therefore
\begin{equation} \label{eq:two}
v(B) + v(B - y - z + w + x) \: \leq \: v(B - y + x) + v(B - z + w) .
\end{equation}
Now, by Equations~(\ref{eq:one}) and~(\ref{eq:two}):
\[
v(A) + v(B) \: = \: v(A) + v(B - z + w) + v(B) - v(B - z + w) \: \leq \:
\]
\[
v(A - x + y) + v(B - y - z + w + x) + v(B - y + x) - v(B - y - z + w + x) \: = \:
\]
\[
v(A - x + y) + v(B - y + x).
\]
\end{proof}

\section{Two more characteristic properties} \label{sec:LLN}
We shall now introduce two more properties that will be shown to be equivalent to 
{\bf Exchange} and {\bf Ultra} in Theorem~\ref{the:equivalence}.
Both are named by the initials of the authors of the articles in which they were presented.
The following is obvious from the definition of marginal valuations.
\begin{lemma} \label{the:LLN-RGP}
The two following properties of a valuation $v$ on $\Omega$ are equivalent:
\begin{enumerate}
\item {\bf LLN}
for any \mbox{$S \subseteq \Omega$}
and for any pairwise distinct \mbox{$x , y , z \in \Omega - S $} 
\[
v_{S}(z \mid x) \: > \: v_{S}(z \mid y) \ \Rightarrow \ v_{S}(x \mid y) \: \geq \: 
v_{S}(x \mid z),
\]
\item {\bf RGP}
for any \mbox{$S \subseteq \Omega$}
and for any pairwise distinct \mbox{$x , y , z \in \Omega - S $} 
\[
v_{S}(x + y) + v_{S}(z) \: \leq \: \max( v_{S}(x + z) + v_{S}(y) , v_{S}(y + z) + v_{S}(x) ).
\]
\end{enumerate}
\end{lemma}

Property {\bf LLN} was proved to be satisfied by substitutes valuations in~\cite{LLN:GEB}
(Claim 1 and Lemma 4) and used to show that substitutes valuations have measure zero
in the set of all valuations.
Property {\bf RGP} was described in~\cite{RGP:ET} (Theorem 10). 
The authors show that substitutes
valuations are exactly those submodular valuations that satisfy {\bf RGP}.

\begin{lemma} \label{the:LLN}
Any valuation $v$ that satisfies {\bf LLN} satisfies {\bf Ultra}.
\end{lemma}
\begin{proof}
Assume \mbox{$c_{A} ( x , y ) \: > \:$} \mbox{$c_{A} (x , z )$}. 
We have
\mbox{$v_{A} ( x + y ) - v_{A} ( x ) - v_{A} ( y ) \: > \:$}
\mbox{$v_{A} (x + z ) - v_{A} ( x ) - v_{A} ( z )$} and
\mbox{$v_{A} ( x \mid y) \: > \:$} \mbox{$v_{A} ( x \mid z )$}.
{\bf LLN}, then implies \mbox{$v_{A} ( y \mid z ) \: \geq \:$}
\mbox{$v_{A} ( y \mid x )$}, and therefore
\mbox{$v_{A} ( y + z ) - v_{A} ( z ) \: \geq \:$}
\mbox{$v_{A} ( x + y ) - v_{A} ( x)$},
i.e., \mbox{$c_{A} ( y , z ) \: \geq \:$} \mbox{$c_{A} ( x , y )$}.

But if we had \mbox{$c_{A} ( y , z ) \: > \:$} \mbox{$c_{A} ( x , y )$},
by the reasoning just above, exchanging $x$ and $y$, we would conclude that 
\mbox{$c_{A} ( x , z ) \: \geq \:$} \mbox{$c_{A} ( x , y )$}, contradicting our assumption.
Therefore \mbox{$c_{A} ( y , z ) \: = \:$} \mbox{$c_{A} ( x , y )$}.
\end{proof}

\section{A third characteristic property} \label{sec:DT}
We shall now discuss an additional property, also named after its inventors.
Theorem~\ref{the:equivalence} will show that it is also equivalent to {\bf Exchange}.
\begin{definition} \label{def:DT}
Let $v$ be any valuation on $\Omega$.
If, for any \mbox{$S \subseteq \Omega$}, any \mbox{$T \subseteq \Omega - S$},
\mbox{$\mid T \mid \: \geq \: 3$} and
any \mbox{$x \in T$} there is some \mbox{$y \in T$}, \mbox{$y \neq x$} such that
\[
{\bf DT} \ \ v(S + x) + v(S + T - x) \: \leq \: v(S + y) + v(S + T - y)
\]
we shall say that $v$ satisfies the property {\bf DT}.
\end{definition}
Property {\bf DT} was shown in~\cite{DT:AML} (Theorem 3) to be a characteristic property 
of those valuations for which a bundle of maximal value can be found by a greedy algorithm.
The present work shows that these valuations are exactly the ultra valuations.

\begin{lemma} \label{the:DT}
The property {\bf DT} implies {\bf LLN} and is implied by {\bf Exchange}.
\end{lemma}
\begin{proof}
Property {\bf RGP} is exactly property {\bf DT} when \mbox{$\mid T \mid = 3$}.
Therefore we may conclude by Lemma~\ref{the:LLN-RGP} that {\bf DT} implies {\bf LLN}.
Property {\bf DT} follows from {\bf Exchange} by taking \mbox{$A = S + x$} and
\mbox{$B = S + T - x$} and noticing that \mbox{$\mid S + x \mid \: \leq \:$}
\mbox{$\mid S + T - x \mid$} since $T$ contains at least three items.
In fact \mbox{$\mid S + x \mid \: < \:$}
\mbox{$\mid S + T - x \mid$}.
\end{proof}

We conclude that all the properties considered so far are equivalent.
\begin{theorem} \label{the:equivalence}
Properties {\bf Exchange, Ultra, LLN, RGP} and {\bf DT} are equivalent.
\end{theorem}
\begin{proof}
By Lemmas~\ref{the:central}, \ref{the:LLN-RGP}, \ref{the:LLN} and \ref{the:DT}.
\end{proof}

\section{Examples of ultra valuations} \label{sec:examples}
Theorem 10 of~\cite{RGP:ET} shows that substitutes valuations satisfy {\bf RGP} and, 
independently, Claim 1 and Lemma 4 of~\cite{LLN:GEB} show that substitutes valuations
satisfy {\bf LLN}, therefore any substitutes valuation is an ultra valuation, 
but not all ultra valuations are substitutes and we shall present two examples 
of families of ultra valuations that exhibit complementarities.
But, before that, note that the {\bf LLN} property shows that ultra valuations 
lie in the union of three hyperplanes in the ${2}^{n}$ dimensional Euclidean space and 
have therefore measure zero, as proved in Theorem 7 of~\cite{LLN:GEB}.
Therefore, the set of ultra valuations is, in a sense, {\em small}.

\subsection{Symmetric valuations} \label{sec:symmetric}
One easily sees that any symmetric valuation $v$, i.e., $v(A)$ depends only on $\mid A \mid$
is an ultra valuation. Symmetric valuations may exhibit complementarities.
If \mbox{$n \: = \: \mid \Omega \mid$} a symmetric valuation is defined 
by $n + 1$ parameters, \mbox{$v^{0} , v^{1} , \ldots , v^{n}$} where 
\mbox{$v(A) \: =\;$} \mbox{$v^{i}$} for any bundle $A$ of size $i$. 
But a larger family of ultra valuations whose elements are defined by $2 n$ parameters can
be considered by giving possibly a different value to each singleton.
For the parameters \mbox{$v^{0} , v^{2} , \ldots , v^{n}$} and 
\mbox{$v_{x}$} for every \mbox{$x \in \Omega$}, let \mbox{$v^{1} \: = \:$} $0$ and
define
\mbox{$v(A) \: = \:$} \mbox{$\sum_{x \in A} v_{x} + v^{k}$} for any
\mbox{$A \subseteq \Omega$}, where
\mbox{$k \: = \:$} \mbox{$\mid A \mid$}.
One easily sees that 
\mbox{$v_{A}(x \mid y) \: = \:$} \mbox{$v_{x} + v^{k + 2} - v^{k + 1} \: = \:$}
\mbox{$v_{A}(x \mid z)$}.

\subsection{Left and Right} \label{sec:LR}
Now, we shall present an interesting family of ultra valuations, formalizing the paradigmatic
example of complementarities: left and right shoes.
Let \mbox{$\Omega \: = \: L \cup R$} with \mbox{$L \cap R \: = \: \emptyset$}.
Define  \mbox{$ pairs(X) \: =\:$}
\mbox{$\min ( \mid X \cap L \mid , \mid X \cap R \mid ) $}, i.e., the number of pairs
\mbox{$(a , b )$}, with \mbox{$a \in L$} and \mbox{$b \in R$} that one can obtain from $X$.
Let \mbox{$v(X) \: = \: $}
\mbox{$f ( pairs(X) ) $} for any {\em strictly increasing concave} function 
\mbox{$f : \{ 0 , \ldots , n \} \rightarrow \cR$}, i.e., \mbox{$x \geq y$} iff 
\mbox{$f(x) \geq f(y)$} and \mbox{$f(m) \: \geq \: \frac{f(m - 1) + f(m + 1)}{2}$}.
We claim that $v$ is an ultra valuation.
We shall show that $v$ satisfies the {\bf Ultra} property.
Let \mbox{$A \subseteq \Omega$} and let \mbox{$x , y , z \in \Omega - A$} 
be pairwise distinct.
Note that \mbox{$v_{A}(x) \: \geq \:$} \mbox{$v_{A}(y)$} if and only if
\mbox{$f(pairs(A + x)) \: \geq \:$} \mbox{$f(pairs(A + y))$}, i.e.,
\mbox{$pairs(A + x) \: \geq \:$} \mbox{$pairs(A + y)$}.
If all three items $x$, $y$ and $z$ are of the same type ($L$ or $R$), then
\mbox{$v_{A}(x) \: = \: $} \mbox{$v_{A} ( y ) \: = \:$} \mbox{$v_{A} ( z )$} and
\mbox{$v_{A} ( x + y ) \: = \:$} \mbox{$v_{A} ( x + z ) \: = \:$} \mbox{$v_{A} ( y + z )$}
and therefore one has \mbox{$c_{A}(x , y) \: = \:$}
\mbox{$ c_{A}(x , z) \: = \:$} \mbox{$ c_{A}(y , z)$}, satisfying {\bf Ultra}.
Otherwise two items are of the same type and the third one is of the opposite type.
Without loss of generality, assume $x$ and $y$ are of the same type and $z$ is of the opposite
type.
We have \mbox{$v_{A}(x) \: = \: $} \mbox{$v_{A} ( y )$} and
\mbox{$v_{A} ( x + z ) \: = \:$} \mbox{$v_{A} ( y + z )$}.
Therefore \mbox{$c_{A}( x , z ) \: = \:$} \mbox{$c_{A}( y , z )$}.
To show that {\bf Ultra} is satisfied we must show that  \mbox{$c_{A}( x , z ) \: \geq \:$} \mbox{$c_{A} ( x , y )$}, i.e.,
\[
v(A + x + z) - v(A) - v(A + x) + v(A) - v(A + z) + v(A) \: \geq \: 
\]
\[
v(A + x + z) - v(A) - v(A + x) + v(A) - v(A + y) + v(A),
\]
i.e.,
\[
v(A + x + z) - v(A + z) \: \geq \: v(A + x + y) - v(A + y).
\]

If $A$ is balanced, i.e., 
\mbox{$\mid A \cap L \mid \: = \:$} \mbox{$\mid A \cap R \mid$} then 
\[
pairs(A) \: = \: pairs(A + y) \: = \: pairs(A + z) \: = \: pairs(A + x + y)
\]
and
\[
pairs(A) + 1 \: = \: pairs( A + x + z ).
\]
Therefore
\[
f(pairs(A + x + z)) - f(pairs(A + z)) \: > \: f(pairs(A + x + y)) - f(pairs(A + y))
\]
and we conclude that indeed
\mbox{$v(A + x + z) - v(A + z) \: \geq \:$} \mbox{$ v(A + x + y) - v(A + y)$}.

If $A$ has more items of the type of $x$ (and $y$) than items of the type of $z$ 
\[
pairs(A) \: = \: pairs(A + y) \: = \: pairs(A + x + y)
\]
and
\[
pairs(A) + 1 \: = \: pairs( A + x + z ) \: = \:  pairs(A + z).
\]
Therefore
\[
f(pairs(A + x + z)) - f(pairs(A + z)) \: = \: 0 \: = \: f(pairs(A + x + y)) - f(pairs(A + y))
\]
and again
\mbox{$v(A + x + z) - v(A + z) \: \geq \:$} \mbox{$ v(A + x + y) - v(A + y)$}.

If $A$ has more items of the type of $z$ we must distinguish two cases.
On one hand, if $A$ has one more item of the type of $z$ than of type $y$ then
\[
pairs(A) + 1 \: = \: pairs(A + x + y) \: = \: pairs( A + x + z ) \: = \:
pairs(A + y)
\]
and 
\[
pairs(A) \: = \: pairs(A + z)
\]
Therefore 
\mbox{$v(A + x + z) - v(A + z) \: > \: $} 
\mbox{$0 \: = \:$} \mbox{$ v(A + x + y) - v(A + y)$}.
On the other hand if $A$ has at least two more items of type $z$ than of type $x$ then
\[
pairs(A) + 2 \: = \: pairs(A + x + y),
\]
\[
pairs(A) + 1 \: = \: pairs(A + y) \: = \: pairs(A + x + z)
\]
and
\[
pairs(A) \: = \: pairs(A + z).
\]
Let \mbox{$k \: = \:$} \mbox{$pairs(A) + 1$}.
We have
\[
f(pairs(A + x + z)) - f(pairs(A + z)) \: = \: f(k) - f(k - 1) \: \geq \: 
\]
\[
f(k + 1) - f(k) \: = \:
f(pairs(A + x + y)) - f(pairs(A + y )).
\]
and again \mbox{$v(A + x + z) - v(A + z) \: \geq \:$} \mbox{$ v(A + x + y) - v(A + y)$}.

The reader is invited to show that, with respect to the left-right partition, the valuation $v$ 
above satisfies the {\em gross substitutes and complements} condition 
of Sun and Yang~\cite{Sun_Yang:econometrica}.
Therefore an economy in which all agents have such valuations has a Walrasian equilibrium.

The example above concerned with items of two different types cannot be extended 
to three types of items: 
if $x$, $y$ and $z$ are of pairwise different types and $w$ is the type of $z$ 
\mbox{$c_{w}(x , y) \: > \:$}
\mbox{$c_{w} (x , z) \: = \:$} \mbox{$c_{w}( y , z)$} contradicting {\bf Ultra}.

\section{Closure properties} \label{sec:closure}
As noticed in Lema~\ref{the:marginal} if $v$ is an ultra valuation and $p$ is a price vector then the utility $u$ defined by
\mbox{$u^{p}(A) \: = \:$} \mbox{$ v(A) - \sum_{i \in A} p_{i}$} is also an ultra valuation.

In contrast with substitutes valuations, 
ultra valuations are {\em not} closed under OR (i.e., convolution), as shown below.
Let $v$ be defined by \mbox{$v(\emptyset) = 0$}, \mbox{$v(x) = 5$}, 
\mbox{$v(y) = 10$}, \mbox{$v(z) = 15$}, \mbox{$v(xy) = 25$}, \mbox{$v(xz) = 20$},
\mbox{$v(yz) = v(xyz) = 35$}. Note that \mbox{$c(x , y) = 10$}, \mbox{$c(y , z) = 10$}
and \mbox{$c(x , z) = 0$}. The valuation $v$ satisfies {\bf Ultra} and is an ultra valuation.
Let $u$ be the symmetric valuation defined by: \mbox{$0 \rightarrow 0$}, 
\mbox{$1 \rightarrow 6$}, \mbox{$2 \rightarrow 12$} and \mbox{$3 \rightarrow 12$}.
The valuation $u$ is symmetric and therefore an ultra valuation.
Let \mbox{$w = u OR v$}.
We have \mbox{$w(x) = 6$}, \mbox{$w(y) = 10$}, \mbox{$w(z) = 15$}, 
\mbox{$w(xy) = 25$}, \mbox{$w(xz = 21$}, \mbox{$w(yz) = 35$} and
\mbox{$c(x , y ) = 9$}, \mbox{$c(y , z) = 10$} and \mbox{$c(x , z) = 0$}, 
that contradicts {\bf Ultra} and shows $w$ is not an ultra valuation.

Ultra valuations are not closed under XOR either: for $u$ and $v$ as above, 
the valuation \mbox{$ ( u \: XOR \: v )$} is not an ultra valuation.

The sum of different ultra valuations on disjoint sets of items is an ultra valuation.
More precisely if \mbox{$\Omega_{0} , \ldots , \Omega_{m}$} is a partition of $\Omega$
and if $v_{i}$ is an ultra valuation on $\Omega_{i}$, the valuation 
\mbox{$v \: = \:$} \mbox{$ \sum_{i = 0}^{m} v_{i}$} is an ultra valuation: $v$ 
satisfies the {\bf Exchange} property.
The example of Section~\ref{sec:examples}  fits the situation of a shop selling shoes of only one size and one make,
but it can be extended to a shop selling different makes and sizes.
Every pair of left and right shoes of the same make and the same size has a value that 
depends on the make and perhaps the size.
The value of a bundle is the sum of the values of the different bundles of pairs of a specific make
and size. This valuation is the sum of the different ultra valuations on each subdomain
defined by a make and a size and is therefore an ultra valuation.

Ultra valuations are {\em not} closed under the addition of a dummy item to $\Omega$.
A study of the effect of dummy items can be found in Appendix~\ref{app:dummy}.

\section{A characterization of substitutes valuations} \label{sec:ultra_sub}
In~\cite{RGP:ET} Theorem 10 establishes that substitutes valuations are exactly those
submodular valuations that satisfy {\bf RGP}.
An alternative proof is proposed below.

\begin{theorem} \label{the:ultra_sub}
A valuation $v$ is substitutes iff
\begin{enumerate}
\item $v$ is submodular, and
\item $v$ is an ultra valuation.
\end{enumerate}
\end{theorem}
\begin{proof}
Gul and Stacchetti~\cite{GulStacc:99} proved that any substitutes valuation is submodular.
Lehmann, Lehmann and Nissan~\cite{LLN:GEB}'s Claim 1, together with the folklore
remark that
the conditional valuation $v_{A}$ is substitutes if $v$ is substitutes, 
proves that any substitutes valuation is an ultra valuation.

Suppose now that $v$ is a submodular ultra valuation, we want to prove that it is
substitutes.
We shall prove that any submodular ultra valuation is $M^{\natural}$-concave
(see Definition~\ref{def:M-natural} below)
and that any such valuation is substitutes.
Murota and Shioura introduced $M^{\natural}$-concave valuations in~\cite{MurotaShioura:99}
in the context of convex discrete optimization.
\begin{definition} \label{def:M-natural}
A valuation $v$ is said to be $M^{\natural}$-concave iff 
for any \mbox{$A , B \subseteq \Omega$} and any \mbox{$x \in A - B$} one of the two
following properties holds:
\begin{enumerate}
\item \mbox{$v(A) + v(B) \: \leq \: v(A - x) + v(B + x)$}, or
\item there exists some \mbox{$y \in B - A$} such that
\[
v(A) + v(B) \: \leq \: v(A - x + y) + v(B - y + x).
\]
\end{enumerate}
\end{definition}
Note the difference with the Exchange property defining ultra valuations: the condition 
\mbox{$\mid A\mid \: \leq \: \mid B \mid$} is dropped but a second possibility is opened.

The proofs that any submodular ultra valuation is $M^{\natural}$-concave and that 
any $M^{\natural}$-concave valuation is substitutes appear in Appendix~\ref{app:substitutes}.
The latter is the easy part of a result of Fujishige and Yang~\cite{FujiYang:03} showing
that a valuation is substitutes iff it is $M^{\natural}$-concave.
\end{proof} 

\section{Preferred bundles} \label{sec:preferred}
For substitutes valuations there is an important literature on the relations between 
preferred bundles, i.e., 
on the properties of the demand correspondence, initiated 
by Gul and Stacchetti~\cite{GulStacc:99} and their characterization of substitutes valuations
by the single improvement property that presents the advantage of not involving prices.
Since ultra valuations are not necessarily substitutes, they do not always satisfy the
single improvement property. 
We shall propose weaker properties.

For any natural number \mbox{$k \: \leq \: n$} ($n$ is the size of $\Omega$)
we shall call the set of bundles of size $k$ 
the $k$-slice of $\Omega$.
A bundle that maximizes a valuation $v$ among all bundles of the $k$-slice is called
a $k$-preferred bundle.

Our first result examines the relation between two preferred bundles (in the same slice or
in different slices).
\begin{theorem} \label{the:2-k}
Let $v$ be an ultra valuation and let \mbox{$A , B \subseteq \Omega$} be bundles of size
$k_{1}$ and $k_{2}$ respectively with 
\mbox{$k_{1} \: \leq \:$} \mbox{$ k_{2}$}.
Assume further that $A$ is a $k_{1}$-preferred bundle and 
that $B$ is a $k_{2}$-preferred bundle.
If \mbox{$x \in A - B$} then there exists some \mbox{$y \in B - A$} such that
the bundle \mbox{$A - x + y$} is a $k_{1}$-preferred bundle 
and \mbox{$B - y + x$} is a $k_{2}$-preferred bundle.
\end{theorem}
\begin{proof}
Since \mbox{$\mid A \mid \: \leq \:$} \mbox{$\mid B \mid$} {\bf Exchange} implies
the existence of a \mbox{$y \in B - A$} such that
\mbox{$v(A) + v(B) \: \leq \:$} \mbox{$v(A - x + y) + v(B - y + x)$}.
But, by our assumptions, \mbox{$v(A) \: \geq \:$} \mbox{$v(A - x + y)$} and
\mbox{$v(B) \: \geq \:$} \mbox{$v(B - y + x)$} and therefore
\mbox{$v(A - x + y) \: = \:$} \mbox{$v(A)$} and 
\mbox{$v(B - y + x) \: = \:$} \mbox{$v(B)$}.
\end{proof}

Our next result is the basis for proving the correctness of the greedy algorithm for finding
an optimal bundle.
\begin{theorem} \label{the:co-greedy}
Let $v$ be an ultra valuation, \mbox{$0 \: < \: k \: < \: n$}
and let $A$ be a $k$-preferred bundle.
Then
\begin{enumerate}
\item  there is an item \mbox{$x \in \Omega - A$} such that $A + x$ 
is a $k + 1$-preferred bundle, and
\item there is an item \mbox{$x \in A$} such that $A - x$ is a $k - 1$-preferred bundle.
\end{enumerate}
\end{theorem}
\begin{proof}
Let $B$ be one of the $k + 1$-preferred bundles that is closest to $A$,
i.e., \mbox{$\mid B - A \mid$} is minimal.
If there exists some item \mbox{$x \in A - B$}, then Theorem~\ref{the:2-k} implies
that there is some \mbox{$y \in B - A$} such that $B - y + x$ is a $k + 1$-preferred bundle. 
But $B - y + x$ is strictly closer to $A$ than $B$, contradicting our assumption.
We conclude that there is no such item $x$ and \mbox{$A \subseteq B$}.
Similarly for the second claim: if $B$ is one of the $k - 1$-preferred bundles
that is closest to $A$, there can be no \mbox{$x \in B - A$} and \mbox{$B \subseteq A$}.
\end{proof}

\begin{corollary} \label{the:kk+1}
Let $v$ be an ultra valuation, $S$ a $k$-preferred bundle
and \mbox{$x \: = \:$} \mbox{${\rm argmax}_{x \in \Omega - S} v(x \mid S)$}.
Then, \mbox{$S + x$} is a $k+1$-preferred bundle.
\end{corollary}
\begin{proof}
By Theorem~\ref{the:co-greedy}.
\end{proof}

\begin{theorem} \label{the:greedy}
If $v$ is an ultra valuation and $A$ a $k$-preferred bundle, then
\begin{enumerate}
\item for any $m$, \mbox{$0 \: \leq \: m \: < \: k$} there is an $m$-preferred bundle $B$ 
such that \mbox{$B \subset A$}, and
\item for any $m$ \mbox{$ k \: < \: m \: \leq \: n$} there is an $m$-preferred bundle $B$ 
such that \mbox{$A \subset B$}.
\end{enumerate}
\end{theorem}
\begin{proof}
By a repeated use of Theorem~\ref{the:co-greedy}.
\end{proof}

We shall say that \mbox{$A \subseteq \Omega$} is a preferred bundle (for $v$) iff 
it maximizes $v$ over all bundles, i.e., \mbox{$v(A) \: \geq \:$}
\mbox{$v(B)$} for any \mbox{$B \subseteq \Omega$}.

\begin{theorem} \label{the:all-pref}
Let $v$ be an ultra valuation and let \mbox{$A , B \subseteq \Omega$} be preferred bundles
with \mbox{$k_{1} \: = \:$} \mbox{$\mid A \mid \: \leq \:$} \mbox{$k_{2} \: = \:$}
\mbox{$ \mid B \mid$}.
Then, 
\begin{enumerate}
\item there is a preferred bundle $C$ such that \mbox{$\mid C \mid \: = \:$} $k_{1}$
and \mbox{$C \subseteq B$}, and
\item  there is a preferred bundle $D$ such that \mbox{$\mid D \mid \: = \:$} $k_{2}$
and \mbox{$A \subseteq D$}, 
\end{enumerate}
\end{theorem}
\begin{proof}
By Theorem~\ref{the:greedy} and the remark that any $k_{i}$-preferred bundle is a preferred
bundle (\mbox{$i = 1 , 2$}).
\end{proof}

\section{Item prices} \label{sec:prices}
The interest in substitutes valuations has been motivated by their properties under a change
in the prices of the items.
We shall now study how preferred sets of an ultra valuation react to a change in prices.
Given a price vector $p$ a preferred bundle for $v$ under prices $p$ is a preferred bundle
for the valuation \mbox{$u_{p}$} defined by 
\mbox{$u_{p}(A) \: = \:$} \mbox{$v(A) - \sum_{x \in A} p_{x}$}.
We noticed in Section~\ref{sec:closure} that $u_{p}$ is an ultra valuation if $v$ is.
Our first result considers an increase in the price of a single item.

\begin{theorem} \label{the:p-k}
Let $v$ be an ultra valuation, $A$ a bundle of size $k$ and \mbox{$x \in A$}.
Let $p$ a price vector for the items and let $p'$ be a price vector such that 
\mbox{$p'_{y} \: = \:$} \mbox{$p_{y}$} 
for any \mbox{$y \in \Omega - \{ x \}$} and \mbox{$p'_{x} \: \geq \:$} \mbox{$ p_{x} $}.
If $A$ is a $k$-preferred bundle at prices $p$, 
then either
\begin{enumerate}
\item $A$ is a $k$-preferred bundle at prices $p'$, or
\item there is some \mbox{$ y \in \Omega - A$} such that \mbox{$A - x + y$} 
is a $k$-preferred bundle at prices $p'$.
\end{enumerate}
\end{theorem}
\begin{proof}
When the price of $x$ increases the utility of all bundles containing $x$ decreases equally, 
whereas the utility of all bundles not containing $x$ stays unchanged.
Bundle $A$ stays preferred among bundles of size $k$ until its utility equals that of a bundle
of size $k$ that does not contain $x$. 
Assume this happens at prices $p_{0}$.
Since $u_{p_{0}}$ is an ultra valuation,
by Theorem~\ref{the:2-k} there is a $k$-preferred bundle of the form \mbox{$A - x + y$}
for some \mbox{$y \in \Omega - A$}. 
A further increase in the price of $x$ does not affect the utility of such a bundle that stays
preferred among bundles of size $k$.
\end{proof}

A similar result holds when the price of $x$ decreases. The proof is similar.
\begin{theorem} \label{the:p-k-d}
Let $v$ be an ultra valuation, $A$ a bundle of size $k$ and \mbox{$x \in \Omega - A$}.
Let $p$ a price vector for items and let $p'$ be a price vector such that 
\mbox{$p'_{y} \: = \:$} \mbox{$p_{y}$} 
for any \mbox{$y \in \Omega - \{ x \}$} and \mbox{$p'_{x} \: \leq \:$} \mbox{$ p_{x} $}.
If $A$ is a $k$- preferred bundle at prices $p$, 
then either
\begin{enumerate}
\item $A$ is a $k$-preferred bundle at prices $p'$, or
\item there is some \mbox{$ y \in A$} such that \mbox{$A - y + x$} is a $k$-preferred
bundle at prices $p'$.
\end{enumerate}
\end{theorem}

Our next result is concerned with preferred bundles.
\begin{theorem} \label{the:p-all}
Let $v$ be an ultra valuation, \mbox{$A \subseteq \Omega$} and \mbox{$x \in A$}.
Let $p$ a price vector for items and let $p'$ be a price vector such that 
\mbox{$p'_{y} \: = \:$} \mbox{$p_{y}$} 
for any \mbox{$y \in \Omega - \{ x \}$} and \mbox{$p'_{x} \: \geq \:$} \mbox{$ p_{x} $}.
If $A$ is a preferred bundle at prices $p$, then either
\begin{enumerate}
\item $A$ is a preferred bundle at prices $p'$, or
\item there is no preferred bundle at prices $p'$ that contains $x$ and 
 there is a bundle \mbox{$B \subseteq A - x$} such that either
\begin{enumerate}
\item $B$ is a preferred bundle at prices $p'$, or
\item there is some \mbox{$ y \in \Omega - B$}, \mbox{$y \neq x$} 
such that \mbox{$B + y$} is a preferred bundle at prices $p'$.
\end{enumerate}
\end{enumerate}
\end{theorem}
\begin{proof}
Let \mbox{$\mid A \mid \: = \: $} $k_{1}$.
When the price of $x$ increases the utility of all bundles containing $x$ decreases equally, 
whereas the utility of all bundles not containing $x$ stays unchanged.
Bundle $A$ stays preferred until its utility equals that of a bundle $X$
that does not contain $x$. 
Let \mbox{$\mid X \mid \: = \:$} $k_{2}$.
Assume this happens at prices $p_{0}$.
If \mbox{$k_{2} \: \geq \:$} $k_{1}$, since $u_{p_{0}}$ is an ultra valuation,
by Theorem~\ref{the:all-pref} there is a preferred bundle $Y$ such that
\mbox{$Y \subseteq X$} and \mbox{$\mid Y \mid \: = \: $} \mbox{$\mid A \mid$}.
We see that \mbox{$x \in A - Y$} and by Theorem~\ref{the:2-k} there is a preferred bundle
of the form \mbox{$A - x + y$} for some \mbox{$y \in Y - A$}. 
Our claim is satisfied with \mbox{$B \: = \: A - x$}.
If \mbox{$k_{2} \: < \: $} $k_{1}$, by Theorem~\ref{the:all-pref} there is a preferred
bundle of size $k_{2}$, say $Y$ such that \mbox{$Y \subseteq A$}. 
If \mbox{$x \not \in Y$}, we can take \mbox{$B \: = \: Y$}.
If \mbox{$x \in Y$}, \mbox{$x \in Y - X$} and by Theorem~\ref{the:2-k},  
since $u_{p_{0}}$ is an ultra valuation, there is a preferred bundle of the form 
\mbox{$Y - x + y$} for some \mbox{$y \in X - Y$}. We can take \mbox{$B = Y - x$}.
\end{proof}

Note that, if $v$ is substitutes Theorem~\ref{the:p-all} can be strenghtened by modifying
the condition \mbox{$B \subseteq A - x$} to \mbox{$B \: = \: A - x$}.
The difference is due to the fact that ultra valuations may exhibit complementarity and an
increase in the price of item $x$ may induce an agent to let go of other items 
that complement $x$. 
Note that, in Theorem~\ref{the:p-all}, 
the set $A - B$ is a bundle that exhibits complementarity.
A theorem similar to Theorem~\ref{the:p-all} can be formulated for a decrease in the
price of an item $x$.

We shall now present a converse to Theorem~\ref{the:co-greedy}.
\begin{theorem} \label{the:converse}
Let $v$ be a valuation that is {\em not} an ultra valuation.
There is a price vector $p$, a number $k$, \mbox{$0 < k < n$} 
and a $k$-preferred bundle $A$ under prices $p$ such that for no item 
\mbox{$x \in \Omega - A$} is \mbox{$A + x$} a $k + 1$-preferred bundle.
\end{theorem}
\begin{proof}
The valuation $v$ does not satisfy {\bf Ultra} and therefore there is a bundle $S$ and 
pairwise distinct items $x$, $y$ and $z$ in \mbox{$\Omega - S$} such that
\mbox{$c_{S}(x , y) \: > \:$} \mbox{$ c_{S}(x , z)$} and 
\mbox{$c_{S}(x , y) \: > \:$} \mbox{$ c_{S}(y , z)$}.
We let \mbox{$k \: = \:$} \mbox{$\mid S \mid + 1$}.
We fix the prices of the items in $S$ low enough to ensure that any $k$-preferred
and any $k+1$-preferred bundles include $S$.
We fix the prices of all items of \mbox{$\Omega - S - x - y - z$} to values high enough
to ensure none of those items are contained in any $k$-preferred or $k+1$-preferred 
bundles.
We complete the description of our price schedule
by \mbox{$p_{x} \: = \:$} \mbox{$v(x \mid S)$}, 
\mbox{$p_{y} \: = \:$} \mbox{$v(y \mid S)$} and
\mbox{$p_{z} \: = \:$} \mbox{$v(z \mid S)$}.
We see that \mbox{$u_{p}(S + x) \: = \:$} \mbox{$u_{p}(S + y) \: = \: $}
\mbox{$u_{p}(S + z) \: = \:$} \mbox{$u_{p}(S)$}.
We notice that \mbox{$S + z$} is a $k$-preferred bundle under prices $p$.
Let us examine the $k + 1$ slice. 
The only candidates to being $k + 1$-preferred bundles under prices $p$ are the
three sets \mbox{$S + x + y$}, \mbox{$S + x + z$} and \mbox{$S + y + z$}.
But \mbox{$u_{p}(S + x + y) \: = \: $}
\mbox{$u_{p}(S) + c_{S}(x , y) + v(x \mid S) + v(y \mid S) - p_{x} - p_{y} \: = \:$}
\mbox{$u_{p}(S) + c_{S}(x , y)$}.
Similarly for \mbox{$u_{p}(S + x + z)$} and \mbox{$u_{p}(S + y + z)$}.
We see that \mbox{$u_{p}(S + x + y)$} is strictly larger than \mbox{$u_{p}(S + x + z)$}
and \mbox{$u_{p}(S + y + z)$}.
We conclude that \mbox{$S + x + y$} is the only $k+1$-preferred bundle and that there is no
$k+1$-preferred bundle that extends \mbox{$S + z$}.
\end{proof}

\section{Characterization of preferred bundles} \label{sec:optimal}
We shall characterize preferred and $k$-preferred bundles of an ultra valuation and show that
any bundle that is, in a sense, locallly preferred is globally preferred.
Our first result deals with $k$-preferred bundles.
\begin{theorem} \label{the:size}
Let $v$ be an ultra valuation, \mbox{$A \subseteq \Omega$}, \mbox{$k = \mid A \mid$}
and assume that \mbox{$v(A) \: \geq \: v(A - x + y)$} for any \mbox{$x \in A$} and any
\mbox{$y \in \Omega - A$}.
Then, $A$ is a $k$-preferred bundle.
\end{theorem}
\begin{proof}
Let $X$ be a $k$-preferred bundle that is closest to $A$ 
in terms of Hamming distance: the distance between two bundles is the size of their symmetric
difference.
Suppose there exists some \mbox{$x \in A - X$}.
Then, by {\bf Exchange}, there exists some \mbox{$y \in X - A$} such that
\[
v(A) + v(X) \: \leq \: v(A - x + y) + v(X - y + x).
\]
But \mbox{$X - y + x$} is strictly closer to $A$ than $X$ and therefore
\mbox{$v(X) \: > \: v(X - y + x)$} and \mbox{$v(A) \: < \: v(A - x + y)$}, a contradiction.
We conclude that there is no such $x$ and therefore \mbox{$A = X$} and $A$ 
is a $k$-preferred bundle.
\end{proof}

Let us now characterize preferred bundles.
\begin{theorem} \label{the:optimal}
Let $v$ be an ultra valuation. A bundle \mbox{$A \subseteq \Omega$} is a preferred bundle 
iff all three conditions below are satisfied:
\begin{enumerate}
\item \label{k} $A$ is a $k$-preferred bundle for \mbox{$k = \mid A \mid$},
\item \label{super}
\mbox{$v(A) \: \geq \: v(C)$} for any bundle $C$ such that \mbox{$A \subseteq C$},
\item \label{sub}
\mbox{$v(A) \: \geq \: v(C)$} for any bundle $C$ such that \mbox{$C \subseteq A$}.
\end{enumerate}
\end{theorem}
Combining the previous two theorems one sees that if a bundle is preferred to all bundles
that can be obtained by discarding a single item and acquiring a single new item, 
to all bundles that can be obtained by acquiring any number of new items 
and to all bundles that can be obtained by discarding any number of items,
then it is a preferred bundle.
\begin{proof}
The {\em only if} part is obvious.
Assume that $A$ satisfies each of the three conditions above and that $B$ is a preferred bundle 
that is one of the closest to $A$ in terms of Hamming distance.
Let \mbox{$k \: = \:$} \mbox{$\mid A \mid$}.
If \mbox{$\mid B \mid \: > \: k$}, by {\bf Exchange}, if there is an \mbox{$x \in A - B$}
there is a \mbox{$y \in B - A$} such that 
\mbox{$v(A) + v(B) \: \leq \:$} \mbox{$v(A - x + y ) + v(B - y + x)$}.
But \mbox{$B - y + x$} is closer to $A$ than $B$ and therefore not a preferred bundle
and \mbox{$v(B - y + x) \: < \:$} \mbox{$v(B)$}.
Therefore \mbox{$v(A) \: < \:$} \mbox{$v(A - x + y)$}, a contradiction to our
assumption~\ref{k}.
We conclude that there is no \mbox{$x \in A - B$}, \mbox{$A \subset B$} and, by our 
assumption~\ref{super}, \mbox{$v(A) \: = \: v(B)$} and $A$ is an optimal bundle.

If \mbox{$\mid B \mid \: < \: k$}, by {\bf Exchange}, if there is an \mbox{$x \in B - A$}
there is a \mbox{$y \in A - B$} such that 
\mbox{$v(B) + v(A) \: \leq \:$} \mbox{$v(B - x + y ) + v(A - y + x)$}.
But \mbox{$B - x + y$} is closer to $A$ than $B$ and therefore not an optimal bundle
and \mbox{$v(B - x + y) \: < \:$} \mbox{$v(B)$}.
Therefore \mbox{$v(A) \: < \:$} \mbox{$v(A - y + x)$}, a contradiction to our 
assumption~\ref{k}.
We conclude that there is no \mbox{$x \in B - A$}, \mbox{$B \subset A$} and, by our 
assumption~\ref{sub}, \mbox{$v(A) \: = \: v(B)$} and $A$ is an optimal bundle.

If \mbox{$\mid B \mid \: = \: k$}, we have \mbox{$v(A) \: = \: v(B)$} 
by our assumption~\ref{k} and $A$ is a preferred bundle.
\end{proof}

If $v$ is substitutes, property~\ref{super} can be weakened to 
\mbox{$v(A) \: \geq \: v(A + x)$} for any \mbox{$x \in \Omega - A$} and
property~\ref{sub} can be weakened to 
\mbox{$v(A) \: \geq \: v(A - x)$} for any \mbox{$x \in A$}.
The reason a stronger condition is needed when dealing with ultra valuations is that 
complementarities may be present.
If items $a$ and $b$ are complementary it may be worthwhile to acquire both $a$ and $b$ 
even though the acquisition of any one of them alone is not attractive: \mbox{$A + a + b$}
may be optimal while both \mbox{$A + a$} and \mbox{$A + b$} are less valuable than $A$.
Similarly, in a situation where \mbox{$a , b \in A$}, one may lose by letting go any one of
them but may profit from relinquishing both.

Note that Theorem~\ref{the:optimal} seems to require an exponential number of checks,
for all subsets and supersets of $A$. 
Section~\ref{sec:greedy} will present a polynomial time
algorithm to find an optimal bundle.

\section{Searching for a preferred bundle} \label{sec:greedy}
We consider the task of finding a preferred bundle, i.e., a bundle that maximizes
an ultra valuation $v$.
In~\cite{DT:AML} the authors show that the greedy algorithm to be described below finds,
in \mbox{$O(n^2)$} steps ($n$ is the size of $\Omega$), 
a preferred bundle iff the valuation satisfies {\bf DT}.
We have shown that this last property is equivalent to {\bf Exchange} and this allows for a
simple proof of the correctness of the greedy algorithm.

The greedy algorithm finds a $k$-preferred bundle, $A_{k}$,  for every 
\mbox{$k = 0 , \ldots , n$}.
The bundle $A_{0}$ is $\emptyset$.
For any \mbox{$k = 0 , \ldots , n - 1$}, 
\[
A_{k + 1} \: = \: A_{k} \cup \{ {\rm argmax}_{x \in \Omega - A_{k}} v(x \mid A_{k}) \}.
\]
Any bundle $A_{l}$ such that \mbox{$v(A_{l}) \: \geq \:$} \mbox{$v(A_{k})$}
for every \mbox{$k = 0 , \ldots , n$} is a preferred bundle by Corollary~\ref{the:kk+1}.
A failure of the greedy algorithm to find a preferred bundle for a valuation $v$ 
implies a failure of Corollary~\ref{the:kk+1} and shows that $v$ is not an ultra valuation.
Therefore ultra valuations are exactly those valuations $v$ for which the greedy algorithm above
finds a preferred bundle for all marginal valuations $v_{S}$.

Note that, letting \mbox{$A_{k} \: = \:$} 
\mbox{$A_{k - 1} \cup \{ x_{k} \}$} and \mbox{$A_{k + 1} \: = \:$} 
\mbox{$A_{k} \cup \{ x_{k + 1} \}$}
if $v$ is substitutes, one has
\[
v(A_{k + 1}) - v(A_{k}) \: = \: v(x_{k + 1} \mid A_{k} ) \: \leq \:  
\]
\[
v(x_{k + 1} \mid A_{k - 1} )
\: \leq \: v( x_{k} \mid A_{ k - 1 } ) \: = \: v(A_{k}) - v(A_{k - 1})
\] 
for every
\mbox{k = 1 , \ldots , n - 1} and one may stop the search as soon as 
\mbox{$v(A_{l + 1}) \: \leq \;$} \mbox{$v(A_{l})$}.
It is this shortened version of our greedy algorithm that is shown to characterize substitutes
valuations in~\cite{RPL:GEB}.
Substitute valuations are exactly those valuations $v$ for which the shortened greedy algorithm 
above finds a preferred bundle for all marginal valuations $v_{S}$.

\section{Competitive equilibrium among ultra agents} \label{sec:Walras}
We want to consider now the allocation problem among agents that exhibit ultra valuations.
Let \mbox{$J = \{ 0 , \ldots , m - 1 \}$} be a set of $m$ agents. 
Agent $j$ has valuation $v_{j}$.
We assume all $v_{j}$'s are ultra valuations and want to study 
the resulting exchange economy.
An allocation is a partition \mbox{$A_{0} , \ldots , A_{m - 1}$} of $\Omega$ in $m$ bundles:
bundle $A_{j}$ is allocated to agent $j$.
Despite the positive results presented in Section~\ref{sec:greedy},
we do not know of a polynomial time algorithm to find the allocation to the agents in $J$ 
that maximizes the social welfare.

We know, from~\cite{GulStacc:99}, that not all economies \mbox{$v_{j} , j \in J$} of
ultra valuations have a competitive equilibrium.
Theorem~\ref{the:Walras} characterizes competitive equilibria among ultra valuations:
instead of having to check that an agent's utility is a maximum over all bundles, 
a local search is sufficient:
it is enough to check this for subsets, supersets of the bundle allocated to an agent and 
bundles lying at distance $2$ from this bundle. 
As explained in Section~\ref{sec:closure} the valuation $u_{j}^{p}$ denotes the utility
of agent $j$ at prices $p$ and is an ultra valuation.

Note that, if for substitutes valuations an ascending (or descending) auction as described
by~\cite{KelsoCraw:82} or~\cite{RPL:GEB} provides a Walrasian equilibrium, this is not
necessarily the case for supra valuations, even if the existence of such an equilibrium is guaranteed.
The reason can be found in Theorem~\ref{the:p-all}: a raise in the price of item $x$ can
lead an agent to let go of some other items in his or her preferred bundle: 
this corresponds to the case the bundle $B$ in the theorem is a strict subset of $A - x$ and, 
in such a case, an item can find itself without any agent interested in it.

\begin{theorem} \label{the:Walras}
A pair \mbox{$( A , p )$} is a Walrasian equilibrium iff for any agent $j$ the three following 
conditions hold:
\begin{enumerate}
\item \label{superw}
\mbox{$u_{j}^{p}(A_{j}) \: \geq \: u_{j}^{p}(C)$} for any bundle $C$ such that 
\mbox{$A_{j} \subseteq C$},
\item \label{subw}
\mbox{$u_{j}^{p}(A_{j}) \: \geq \: u_{j}^{p}(C)$} for any bundle $C$ such that 
\mbox{$C \subseteq A_{j}$},
\item \label{sizew}
\mbox{$u_{j}^{p}(A_{j}) \: \geq \: u_{j}^{p}(A_{j} - x + y)$} for any 
\mbox{$x \in A_{j}$} and any \mbox{$y \in \Omega - A_{j}$}.
\end{enumerate}
\end{theorem}
\begin{proof}
The {\em only if} part follows straightforwardly from the definition of a Walrasian equilibrium,
the {\em if part}  from Theorem~\ref{the:optimal}.
\end{proof}


Theorem~\ref{the:Walras} implies that, in an economy of ultra valuations in which 
no Walrasian equilibrium exists, no transactions based on prices between the agents 
can attain an allocation that satisfies the three conditions above.

If there is a Walrasian equilibrium the Linear Program considered in~\cite{BikhMamer:equ}
has an integral solution and it can be computed in polynomial time as explained 
in~\cite{NisanSegal:JET}.
The hard cases for an allocation algorithm therefore concern the case the
economy has no competitive equilibrium and the Linear Program has a fractional solution.

\section{Conclusion and open questions} \label{sec:conclusion}
This paper has proposed a novel exchange property that can be defined in many equivalent 
but quite different ways, some of them considered earlier in the extant literature.
This exchange property allows for streamlined proofs of properties 
of the valuations that satisfy it, the ultra valuations, in particular Theorem~\ref{the:co-greedy} 
that implies a straightforward, polynomial-time, 
greedy algorithm for finding a preferred bundle for an ultra valuation.
The complexity of the problem of finding an optimal allocation is not known.
Ultra valuations may be compared both with substitutes valuations for which both problems 
of maximizing a single valuation and of finding an optimal allocation are easy, 
and with submodular valuations for which both problems are NP-hard.

Here is a list of intriguing open questions.
\begin{itemize}
\item Is the problem of finding an optimal allocation in NP, in P, NP-hard?
\item What is the communication complexity of the allocation problem?
\item What is the complexity of deciding whether a set of $m$ ultra valuations 
admits a competitive equilibrium?
\item What is the notion of equilibrium, i.e. the solution concept,
that could fit economies of ultra agents?
\item A search for more examples of ultra valuations, in particular real life valuations,
is also worthwhile pursuing.
\end{itemize}

\section{Acknowledgments} \label{sec:ack}
The ground work for the results presented here was laid in 2004-2006 in collaboration with
Alejandro Bertelsen and Meir Bing. Discussions with Noam Nisan are gratefully acknowledged.

\appendix 
\section{Technical lemmas} \label{app:lemmas}
\begin{lemma} \label{le:conditional}
Let $v$ satisfy {\bf Ultra} and \mbox{$B \subseteq \Omega$}, then
the valuation $v_{B}$ satisfies {\bf Ultra}.
\end{lemma}
\begin{proof}
The proof is straightforward if one notices that
\mbox{${v_{B}}_{A} = v_{A \cup B}$} and therefore 
\mbox{${c_{B}}_{A} = c_{A \cup B}$}.
\end{proof}

\begin{lemma} \label{le:cv}
Let $v$ satisfy {\bf Ultra} and $x$, $y$, \mbox{$z \in \Omega$} be pairwise distinct items.
One of the two following holds:
\begin{enumerate}
\item \label{one} \mbox{$v(z) + v(x + y) \: \leq \: v(y) + v(x + z)$}, or 
\item \label{two} \mbox{$v(z) + v(x + y) \: = \: v(x) + v(y + z)$}.
\end{enumerate}
\end{lemma}
Note that the consequent is a strengthening of the property {\bf RGP} described in
Section~\ref{sec:LLN}.
\begin{proof}
Inequality~\ref{one} is equivalent to \mbox{$c(x , y) \: \leq \: c(x , z)$} 
and equality~\ref{two} is equivalent to \mbox{$c(y , z) = c(x , y)$}.
The result follows from the definition of {\bf Ultra}.
\end{proof}

\begin{lemma} \label{le:1-n}
Let $v$ satisfy {\bf Ultra},
\mbox{$\emptyset \neq A \subseteq \Omega$} and \mbox{$x \in \Omega - A$}.
Then there is some \mbox{$y \in A$} such that
\mbox{$v(x) + v(A) \: \leq \: v(y) + v(A - y + x)$}.
\end{lemma}
\begin{proof}
The proof proceeds by induction on \mbox{$k \: = \:$} \mbox{$\mid A \mid$}.
The base case is \mbox{$k = 1$} and is obvious.
Suppose we have proved our claim for $k$ and let \mbox{$\mid A \mid \: = \:$} 
\mbox{$k + 1$}.
Let \mbox{$w \in A$} be the item that maximizes the quantity
\mbox{$v(A - z) - v(A - z + x)$} over \mbox{$z \in A$}.
By the induction hypothesis, there is some \mbox{$y \in A - w$} such that
\[
v(x) + v(A - w) \: \leq \: v(y) + v(A - w - y + x)
\]
By Lemma~\ref{le:conditional}  \mbox{$v_{A - w - y}$} satisfies {\bf Ultra} and
Lemma~\ref{le:cv} shows that either
\begin{equation} \label{eq:1-n}
v(A - w - y + x) + v(A) \: \leq \: v(A - y + x) + v(A - w)
\end{equation}
or
\begin{equation} \label{eq:2-n}
v(A - w - y + x) + v(A) \: \leq \: v(A - w + x ) + v(A - y).
\end{equation}
But we have chosen $w$ in such a way that
\[
v(A - w) - v(A - w + x) \: \geq \: v(A - y ) - v(A - y + x)
\] 
and therefore 
\[
v(A - y + x) + v(A - w) \: \geq \: v(A - w + x) + v(A - y).
\]
We conclude that~(\ref{eq:2-n}) implies~(\ref{eq:1-n}) 
and consequently~(\ref{eq:1-n}) holds.
Therefore:
\[
v(x) + v(A) \: = \: v(x) + v(A - w) + v(A) - v(A-w) \: \leq \: 
\]
\[
v(y) + v(A - w - y + x) + v(A - y + x) - v(A - w - y + x) \: = \:
v(y) + v(A - y + x ).
\]
\end{proof}

\begin{lemma} \label{le:2-2}
Let $v$ satisfies {\bf Ultra} and $w$, $x$, $y$, \mbox{$z \in \Omega$} be
pairwise distinct.
One of the two following inequalities holds:
\begin{enumerate}
\item \label{wx} \mbox{$c(w , y) + c(x , z) \: \leq \: c(w , x) + c(y , z)$}, or 
\item \label{xy} \mbox{$c(w , y) + c(x , z) \: \leq \: c(x , y) + c(w , z)$}.
\end{enumerate}
\end{lemma}
\begin{proof}
If none of the inequalities holds, some term on the left must be strictly greater than the
corresponding term on the right. 
All such pairs are the same up to a permutation of the variables, therefore, without
loss of generality we may assume \mbox{$c(w , y) > c(w , x)$}.
By {\bf Ultra} we have:
\[
c(x , y) \: = \: c(w , y) \: > \: c(w , x).
\]
If \mbox{$c(x , z) \: \leq \: c(w , z)$} inequality~\ref{xy} is satisfied.
Assume, then, that \mbox{$c(x , z) \: > \: c(w , z)$}.
By {\bf Ultra} we have:
\[
c(x , y) \: = \: c(w , y) \: > \: c(w , x) \: = \: c( x , z ) \: > \: c(w , z).
\]
Notice that \mbox{$c(w , y) \: > \: c(w , z)$} and therefore, by {\bf Ultra},
we have:
\[
c(x , y) \: = \: c(w , y) \: = \: c(y , z) \: > \: c(w , x) \: = \: c( x , z ) \: > \: c(w , z) 
\]
and inequality~\ref{xy} is satisfied.
\end{proof}

\section{Dummy items} \label{app:dummy}
Let $v$ be a valuation on $\Omega$.
When we say that we add a set $X$, \mbox{$X \cap \Omega =$} $\emptyset$,
of dummy items we mean that we consider the
valuation $v'$ on the set \mbox{$\Omega' \: = \:$} \mbox{$\Omega \cup X$} defined by
\mbox{$v'(A) \: = \:$} \mbox{$v( A \cap \Omega )$} for any \mbox{$A \subseteq \Omega'$}.

\begin{lemma} \label{the:d1}
If $v$ is not submodular then $v'$ does not satisfy {\bf LLN}.
Therefore adding dummy items does not preserve ultra valuations.
\end{lemma}
\begin{proof}
Assume $v$ is not submodular.
There are \mbox{$A \subseteq \Omega$} and \mbox{$x , y \in \Omega - A$}
such that \mbox{$v( x \mid A + y ) \: > \:$} \mbox{$v( x \mid A )$}.
Let \mbox{$z \not \in \Omega $} be a dummy element.
We have 
\[
v'_{A} ( x \mid y ) \: = \: v(x \mid A + y ) \: > \: v( x \mid A ) \: = \: v'(x \mid A + z ) \: = \:
v'_{A} ( x \mid z ).
\]
But 
\mbox{$v'_{A} ( y \mid z ) \: = \:$} \mbox{$ v ( y \mid A ) $},
\mbox{$v'_{A} ( y \mid x ) \: = \:$} \mbox{$ v( y \mid A + x ) $} and
\[
v( y \mid A ) - v(y \mid A + x) \: = \: v ( x \mid A ) - v( x \mid A + y ) \: < \: 0.
\]
\end{proof}

But adding dummy items preserve ultra valuations that are submodular.
\begin{lemma} \label{the:d2}
If $v$ is a submodular ultra valuation then $v'$ is a submodular ultra valuation.
\end{lemma}
\begin{proof}
One easily sees that adding dummy items preserves submodularity.
Suppose $v$ is a submodular ultra valuation, we shall show that $v'$ satisfies {\bf LLN}.
Let \mbox{$A \subseteq \Omega' \: = \:$} \mbox{$\Omega \cup X$},
\mbox{$x , y , z \in \Omega' $}  pairwise distinct and assume that 
\mbox{$v'_{A}( z \mid x ) \: > \:$} \mbox{$v'_{A} ( z \mid y )$}.
We see that $z$ is not a dummy item: \mbox{$z \in \Omega$}.
If $x$ is a dummy item the consequent of {\bf LLN} holds since
\mbox{$v'_{A} ( x \mid y ) \: = \:$} \mbox{$v'_{A} ( x \mid z ) \: = \:$} $0$.
We may therefore assume now that \mbox{$x \in \Omega$}.
If $y$ were a dummy item we would have 
\mbox{$v' ( z \mid A + x ) \: = \:$} \mbox{$v'_{A} ( z \mid x ) \: > \:$} 
\mbox{$v'_{A} ( z \mid y ) \: = \:$} \mbox{$v' ( z \mid A )$},
contradicting the fact that $v'$ is submodular.
We see that we may assume that 
\mbox{$x , y , z \in \Omega$}.
But, if \mbox{$B = A \cap \Omega$}, for any \mbox{$a , b \in \Omega$}
we have \mbox{$v'_{A} ( a \mid b ) \: = \:$}  \mbox{$v_{B} (a \mid b )$}
and our claim follows from the fact that $v$ is an ultra valuation.
\end{proof}

\section{Characterization of substitutes valuations} \label{app:substitutes}
\begin{theorem} \label{the:natural}
Any submodular ultra valuation is $M^{\natural}$-concave.
\end{theorem}
\begin{proof}
Let $v$ be a submodular ultra valuation,
\mbox{$A , B \subseteq \Omega$} and \mbox{$x \in A - B$}.
We must show that either \mbox{$v(A) + v(B) \: \leq \:$} \mbox{$v(A - x) + v(B + x)$},
or there exists some \mbox{$y \in B - A$} such that
\mbox{$v(A) + v(B) \: \leq \:$} \mbox{$v(A - x + y) + v(B - y + x)$}.
If \mbox{$\mid A \mid \: \leq \:$} \mbox{$\mid B \mid$} the claim follows from
the {\bf Exchange} property of $v$.

Suppose, then, that \mbox{$\mid A \mid \: > \:$} \mbox{$\mid B \mid$}.
We shall add dummy items to $B$ in order to bring it to the size of $A$.
Let \mbox{$ k \: = \:$} \mbox{$ \mid A \mid - \mid B \mid \: > \:$} $0$.
Let $X$ be any set of $k$ elements such that 
\mbox{$X \cap \Omega = \emptyset$} and let
\mbox{$\Omega' = \Omega \cup X$} and define $v'$ by
\mbox{$v'(Y) = v(Y \cap \Omega)$} for any \mbox{$Y \subseteq \Omega'$}.
By Lemma~\ref{the:d2}, the valuation $v'$ is an ultra valuation.

Let now \mbox{$B' \: = \:$} \mbox{$B \cup X$}.
We have \mbox{$\mid A \mid \: \leq \:$} \mbox{$\mid B' \mid$} and
\mbox{$x \in A - B'$}.
By the {\bf Exchange} property of $v'$ there is some \mbox{$y \in B' - A$} such that
\mbox{$v'(A) + v'(B') \: \leq \:$} \mbox{$v'(A - x + y) + v'(B' - y + x )$}.
But \mbox{$v'(A) \: = \: v(A)$} and \mbox{$v'(B') \: = \: v(B)$}.
We distinguish two cases: $y$ is a dummy item or an element of $\Omega$.
If \mbox{$y \in \Omega$}, \mbox{$v'(A - x + y) \: = \:$} \mbox{$v(A - x + y)$}
and \mbox{$v'(B' - y + x ) \: = \:$} \mbox{$v(B - y + x)$} and the second option of our
claim is satisfied.
If, on the contrary, $y$ is a dummy item, \mbox{$ y \in X$}, then
\mbox{$v'(A - x + y) \: = \:$} \mbox{$v(A - x)$} and
\mbox{$v'(B' - y + x) \: = \:$} \mbox{$v(B + x)$}
and the first option of our claim is satisfied.
Note that the submodularity of $v$ is used only to prove that $v'$ is an ultra valuation.
\end{proof} 

We are left to show that any $M^{\natural}$-concave valuation is substitutes.
This was shown by Fujishige and Yang in~\cite{FujiYang:03} but, since they use the language 
and techniques of discrete convex optimization an alternative proof is given below.
We shall use the following characterization of substitutes valuations:
\begin{theorem} \label{the:char}
A valuation $v$ is substitutes iff for any vector price $p$, any bundles 
\mbox{$A , B \subseteq \Omega$} both optimal at prices $p$, any \mbox{$x \in A - B$},
there exists a bundle $C$ optimal at prices $p$ such that \mbox{$A - \{ x \} \subseteq C$}
and \mbox{$x \not \in C$}.
\end{theorem}

\begin{theorem} \label{the:M-nat-sub}
Any $M^{\natural}$-concave valuation is substitutes.
\end{theorem}
\begin{proof}
Let $p$, $A$, $B$ and $x$ be as assumed and let $v$ be $M^{\natural}$-concave.
For any bundle \mbox{$X \subseteq \Omega$}, define
\mbox{$u_{p}(X) \: = \:$} \mbox{$v(X) - \sum_{w \in X} p_{w}$}.
If \mbox{$v(A) + v(B) \: \leq \:$} \mbox{$ v(A - x) + v(B + x)$} we have
\mbox{$u_{p}(A) + u_{p}(B) \: \leq \:$} \mbox{$ u_{p}(A - x) + u_{p}(B + x)$}.
But \mbox{$u_{p}(A) \: \geq \:$} \mbox{$ u_{p}(A - x)$}
and \mbox{$u_{p}(B) \: \geq \:$} \mbox{$ u_{p}(B + x)$}.
We conclude that \mbox{$u_{p}(A) \: = \:$} \mbox{$ u_{p}(A - x)$} and
$A - x$ is the bundle $C$ we looked for.

If \mbox{$v(A) + v(B) \: \leq \:$} \mbox{$ v(A - x + y) + v(B - y + x)$} 
for some \mbox{$y \in B - A$}
we have \mbox{$u_{p}(A) + u_{p}(B) \: \leq \:$} 
\mbox{$u_{p}(A - x + y) + u_{p}(B - y + x)$}.
But we have \mbox{$u_{p}(A) \: \geq \:$} \mbox{$ u_{p}(A - x + y)$} and
\mbox{$u_{p}(B) \: \geq \:$} \mbox{$ u_{p}(B - y + x)$}.
We conclude that \mbox{$u_{p}(A) \: = \:$} \mbox{$ u_{p}(A - x + y)$} and
$A - x + y$ is the bundle $C$ we looked for.
\end{proof}

\section{Choice-language definition of substitutability} \label{app:CL}
In~\cite{Hatfield_Milgrom:2005} Hatfield and Milgrom proposed a definition of substitutability
in terms of the relation between the subsets of choice among a set of possibilities 
and among a  larger set of possibilities.
Their definition, intuitively, states that a valuation $v$ is substitutes iff for any bundles 
\mbox{$X \subseteq Y \subseteq \Omega$} and for any bundle \mbox{$A \subseteq X$}
that maximizes $v$ among all subsets of $X$ and any bundle \mbox{$B \subseteq Y$}
that maximizes $v$ among all bundles of $Y$, every item $x$ of $B$ that happens to be
an item of $X$ is already in $A$: \mbox{$B \cap X \subseteq A$}.

The intuitive description above needs to be carefully formalized if one wants to pay attention 
to situations in which the maximizing bundle is not unique.
The authors above, to avoid, this problem, treat only the case 
where there is a unique maximizing bundle, which is enough since this is the generic case.

Theorem~\ref{the:language} formalizes a definition in the style of Hatfield and Milgrom 
in the general case, where ties are permitted, and characterizes the valuations it defines.

\begin{theorem} \label{the:language}
The three following properties of a valuation $v$ are equivalent:
\begin{enumerate}
\item \label{subst} $v$ is substitutes,
\item \label{down} for any additive valuation $p$, \mbox{$X \subseteq Y \subseteq \Omega$} 
and for any \mbox{$B \subseteq Y$} that maximizes $v^{p}$ over all subsets of $Y$, there is an
\mbox{$A \subseteq X$} that maximizes $v^{p}$ over all subsets of $X$ 
such that \mbox{$B \cap X \subseteq A$},
\item \label{up} for any additive valuation $p$, \mbox{$X \subseteq Y \subseteq \Omega$} 
and for any \mbox{$A \subseteq X$} that maximizes $v^{p}$ over all subsets of $X$, 
there is a \mbox{$B \subseteq Y$} that maximizes $v^{p}$ over all subsets of $Y$ 
such that \mbox{$B \cap X \subseteq A$}.
\end{enumerate}
\end{theorem}

\begin{proof}
To show that \ref{subst} implies \ref{down} 
one uses the properties of the substitutes valuation $v$ 
under rising prices. 
Begin with the prices given by $p$ for all items of $Y$ and very high prices for items in
\mbox{$\Omega - Y$}. 
Then raise the prices of all items in \mbox{$Y - X$} to a very high price.

To show that \ref{subst} implies \ref{up} one uses the properties of a substitutes valuation under
decreasing prices. 
Begin with the prices given by $p$ for all items of $X$ and very high prices for items in
\mbox{$\Omega - X$}. 
Then decrease the prices of all items in \mbox{$Y - X$} to their price in $p$.

Let us assume \ref{down} or \ref{up}. First, we shall show that $v$ is submodular.
If \mbox{$A \subseteq B \subseteq \Omega$}, \mbox{$x \in \Omega - B$} and
\mbox{$v( x \mid B ) \, > \,$} \mbox{$ v(x \mid A )$}, 
we choose a price $p_{x}$ for $x$ such that
\mbox{$v( x \mid B ) \, > \,$} \mbox{$p_{x} \, > \,$} \mbox{$ v(x \mid A )$}.
For any item in $B$, choose a very low price (may be negative).
One sees that \mbox{$B \cup \{ x \}$} is the unique bundle that maximizes $v^{p}$ 
among all subsets of \mbox{$B \cup \{ x \}$} and
$A$ is the only bundle that maximizes $v^{p}$ among all subsets of \mbox{$A \cup \{ x \}$}, contradicting both \ref{down} and \ref{up}. 
We conclude that \mbox{$v( x \mid B ) \, \leq \,$} \mbox{$ v(x \mid A )$} and $v$ is submodular.
We shall now show that $v$ satisfies {\bf LLN}, but note we shall use submodularity.
Assume that {\bf LLN} does not hold: 
\mbox{$v_{S}(z \mid x) \, > \,$} \mbox{$v_{S}(z \mid y)$} and 
\mbox{$v_{S}(x \mid y) \, < \,$} \mbox{$v_{S}(x \mid z)$}. 
We shall show that \ref{up} does not hold.
Let us give very low prices to all items in $S$
and fix prices for $x$, $y$ and $z$ such that:
\[
v_{S} (x \mid z ) > p_{x} > v_{S}(x \mid y), \ \ 
v_{S} (z \mid x ) > p_{z} > v_{S}(z \mid y), 
\]
\[
v_{S}(y) - v_{S}(x) + p_{x} > p_{y} > v_{S}(y) - v_{S}(\{ x , z \}) + p_{x} + p_{z}.
\]
This is indeed possible since \mbox{$p_{z} <$} \mbox{$ v_{S}(z \mid x) =$}
\mbox{$v_{S}(\{ x , z \}) - v_{S}(x)$}. 
Treating each inequality in turn, and using the fact $v$ is submodular we get:
\[
v_{S}^{p}(z) < v_{S}^{p}(\{ x , z \}) , \ \ v_{S}^{p}(\emptyset) < v_{S}^{p}(x) , \ \ 
v_{S}^{p}(\{ x , y \}) < v_{S}^{p}(y), \ \ 
v_{S}^{p}(x) < v_{S}^{p}(\{ x , z \}), 
\]
\[
v_{S}^{p}(\emptyset) < v_{S}^{p}(z), \ \ v_{S}^{p}(\{ y , z \}) < v_{S}^{p}(y), \ \ 
v_{S}^{p}(x)  < v_{S}^{p}(y), \ \ v_{S}^{p}(y) < v_{S}^{p}(\{ x , z \}).
\]
We see that the bundle \mbox{$S \cup \{y \}$} is the only bundle that maximizes $v^{p}$ 
among all subsets of \mbox{$S \cup \{x , y \}$} and that 
\mbox{$S \cup \{x , z \})$} is the only bundle that maximizes $v^{p}$ among all subsets of 
\mbox{$S \cup \{x , y , z \}$}, a contradiction to both \ref{up} and \ref{down}.
We conclude that $v$ satisfies {\bf LLN}.
By Theorem~\ref{the:equivalence} it is an ultra valuation.
Since we have shown $v$ is submodular we conclude that  it is substitutes
by Theorem~\ref{the:ultra_sub}.
\end{proof}

\bibliographystyle{plain}
\bibliography{../../../my}

\end{document}